\documentclass[submission,copyright,creativecommons]{eptcs}
\providecommand{\event}{NCL'26} 

\usepackage{xifthen}
\usepackage{comment}

\ifpdf
  \usepackage{underscore}         
  \usepackage[T1]{fontenc}        
\else
  \usepackage{breakurl}           
\fi

\title{Cyclic Nonassociative Bilinear Logic With Additive Constants as Extension of Full Nonassociative Lambek Calculus}
\author{Pawe\l{} P\l{}aczek
\institute{WSB Merito University in Poznan, Poland}
\email{pawel.placzek@poznan.merito.pl}
\and
Cheng-Syuan Wan
\institute{Tallinn University of Technology, Estonia}
\email{cswan@cs.ioc.ee}
}
\def\titlerunning{CyNBL With Additive Constants as Extension of FNL}
\def\authorrunning{P. P\l{}aczek and C.-S. Wan}

\usepackage{amsmath,amssymb,amsthm}
\usepackage{hyperref}
\usepackage{cleveref}
\usepackage[mathscr]{euscript}
\usepackage{xspace}
\usepackage{enumitem}
\usepackage{cmll}
\usepackage{bussproofs}
\usepackage{array,multirow}
\usepackage{hhline}
\usepackage{cite}
\usepackage{graphicx}

\newcommand{\FNLt}{\textup{FNL$_\top$}\xspace}

\newcommand{\rrtensor}{\textup{(r-$\otimes$)}\xspace}
\newcommand{\rrone}{\textup{(r-1)}\xspace}
\newcommand{\roplusone}{\textup{(r-$\oplus$1)}\xspace}
\newcommand{\roplustwo}{\textup{(r-$\oplus$2)}\xspace}
\newcommand{\rwedge}{\textup{(r-$\wedge$)}\xspace}
\newcommand{\rvee}{\textup{(r-$\vee$)}\xspace}
\newcommand{\rshift}{\textup{(r-shift)}\xspace}
\newcommand{\rcyc}{\textup{(r-cyc)}\xspace}
\newcommand{\rcut}{\textup{(r-cut)}\xspace}
\newcommand{\roplusthree}{\textup{(r-$\oplus$3)}\xspace}
\newcommand{\roplusfour}{\textup{(r-$\oplus$4)}\xspace}

\newcommand{\rassmone}{\textup{(r-assm1)}\xspace}
\newcommand{\rassmtwo}{\textup{(r-assm2)}\xspace}

\newcommand{\FNLcut}{\textup{(cut)}\xspace}
\newcommand{\FNLRtotop}{\textup{($\Rightarrow\top$)}\xspace}
\newcommand{\FNLotimesR}{\textup{($\otimes\Rightarrow$)}\xspace}
\newcommand{\FNLRtootimes}{\textup{($\Rightarrow\otimes$)}\xspace}
\newcommand{\FNLmultimapR}{\textup{($\multimap\Rightarrow$)}\xspace}
\newcommand{\FNLRtomultimap}{\textup{($\Rightarrow\multimap$)}\xspace}
\newcommand{\FNLmultimapinvR}{\textup{($\multimapinv\Rightarrow$)}\xspace}
\newcommand{\FNLRtomultimapinv}{\textup{($\Rightarrow\multimapinv$)}\xspace}
\newcommand{\FNLoneR}{\textup{(1$\Rightarrow$)}\xspace}
\newcommand{\FNLRtoone}{\textup{($\Rightarrow$ 1)}\xspace}
\newcommand{\FNLveeR}{\textup{($\vee\Rightarrow$)}\xspace}
\newcommand{\FNLRtovee}{\textup{($\Rightarrow\vee$)}\xspace}
\newcommand{\FNLwedgeR}{\textup{($\wedge\Rightarrow$)}\xspace}
\newcommand{\FNLRtowedge}{\textup{($\Rightarrow\wedge$)}\xspace}

\newcommand{\provS}{\vdash_{\mathbf{S}_\Phi}}
\newcommand{\provFNL}{\vdash_\FNLt}

\newcommand{\agdahen}{%
  \begingroup
  \raisebox{-1mm}
  {\includegraphics[height=1em]{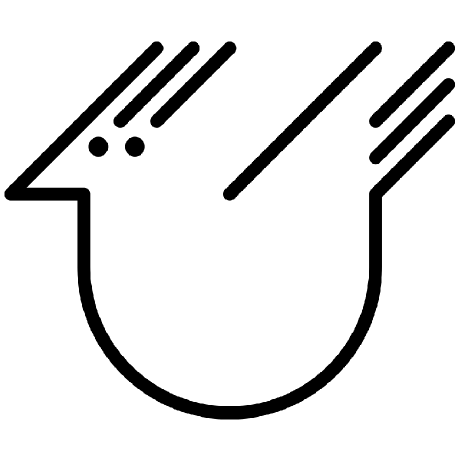}}%
  \endgroup
}
\newcommand{\agdaurl}{https://github.com/cswphilo/CyNBL-conservative-over-FNL-top}
\newcommand{\agdahref}{dbfe8667e652a2461d25700eed688f76d1364730}
\newcommand{\agdaref}[1]{\href{\agdaurl/blob/\agdahref/#1}{\agdahen}}

\theoremstyle{definition}
\newtheorem{definition}{Definition}[section]

\theoremstyle{plain}
\newtheorem{theorem}[definition]{Theorem}
\newtheorem{lemma}[definition]{Lemma}
\newtheorem{corollary}[definition]{Corollary}
\newtheorem{proposition}[definition]{Proposition}

\theoremstyle{remark}
\newtheorem{remark}[definition]{Remark}

\crefname{definition}{Definition}{Definitions}
\crefname{theorem}{Theorem}{Theorems}
\crefname{lemma}{Lemma}{Lemmas}
\crefname{proposition}{Proposition}{Propositions}
\crefname{corollary}{Corollary}{Corollaries}

\newcolumntype{C}{>{\hfil\displaystyle}c<{\hfil}}
\newcolumntype{L}{>{\displaystyle}l<{\hfil}}

\newenvironment{rules}{\[\array{c@{}L@{}c@{}L}\hspace*{0.1\textwidth} & \hspace*{0.4\textwidth} & \hspace*{0.1\textwidth} & \hspace*{0.4\textwidth}\\[-3em]}{\endarray\]\ignorespacesafterend}

\begin{document}
\maketitle

\begin{abstract}
Cyclic Linear Logic (CyMALL, CyBL) is a compromise between Commutative Linear Logic (MALL) and Noncommutative Linear Logic (Bilinear Logic, BL). We do not admit commutativity of the tensor, but we still admit only one De Morgan negation. MALL and its variants are natural extensions of Full Lambek Calculus (FL). Abrusci proved that CyBL is not a conservative extension of FL, if we admit additive constants in these logics. In this paper we show that Cyclic Nonassociative Bilinear Logic (CyNBL) with both additive constants is a strongly conservative extension of Full Nonassociative Lambek Calculus with only one additive constant -- the upper bound -- and that this holds even in the presence of an arbitrary set of nonlogical axioms.
\end{abstract}

\section{Introduction and preliminaries}

In the field of substructural logics two logics have been widely studied: the Lambek Calculus (L) and Linear Logic (LL or PLL). The Lambek Calculus was introduced by Lambek in \cite{lambek-mathematics-1958} as a tool for modeling natural language syntax. Linear Logic was introduced by Girard in \cite{girard-linear-1987} as a refinement of classical logic, which captures the notion of resource sensitivity. Both logics may be considered more like a framework than a single logic -- contrary to classical logic. Usually one considers Multiplicative-Additive Linear Logic (MALL) as the standard version of Linear Logic, which lacks exponentials; in this paper we do so, too.

The two logics at the centre of this paper are Cyclic Nonassociative Bilinear Logic (CyNBL) and Full Nonassociative Lambek Calculus with only the top additive constant (\FNLt), a variant of the Lambek Calculus. Our main result is that CyNBL is a strongly conservative extension of \FNLt. Before stating this precisely, we first explain, in the remainder of this introduction, how CyNBL and \FNLt arise from the broader landscape of substructural logics; readers mainly interested in the result itself may keep this in mind while going through the following historical and terminological remarks.

The connection between Lambek Calulus and Linear Logic is a very important topic in the field of substructural logics. The Lambek Calculus can be seen as a fragment of Linear Logic, and Linear Logic can be seen as an extension of the Lambek Calculus. However, the exact relationship between these two logics is not straightforward, and it has been the subject of much research. First, Lambek Calculus does not admit the exchange rule and hence it is noncommutative, while Linear Logic admits the exchange rule and hence it is commutative. This difference leads to Bilinear Logic (BL, usually called Noncommutative Linear Logic or NMALL) due to Abrusci \cite{abrusci-phase-1991}, which is MALL without the exchange rule. A compromise between MALL and BL is Cyclic Linear Logic (CyBL, also called CyMALL) due to Yetter \cite{yetter-quantales-1990}, which is BL with the cyclic rule, which allows one to cyclically permute the formulas in sequents. The cyclic rule is a weaker version of the exchange rule, so CyBL is noncommutative, but it still admits only one De Morgan negation, so it is not as noncommutative as BL.

On the other hand, the Lambek Calculus can be extended by adding additive connectives and constants, which leads to Full Lambek Calculus (FL). The main difference between FL and BL or CyBL is the presence of negation(s) and a new multiplicative constant. Another interesting variant of L is Nonassociative Lambek Calculus (NL) due to Lambek \cite{lambek-deductive-1968}. The main difference between L and NL is the lack of associativity of the tensor. The extension of NL by additive connectives and constants is Full Nonassociative Lambek Calculus (FNL). Buszkowski \cite{buszkowski-involutive-2017} introduced Involutive Nonassociative Lambek Calculus (InNL)\footnote{In \cite{buszkowski-involutive-2017} the logic does not admit multiplicative constants. But, Buszkowski provided instructions how to add multiplicative constant and additive connectives and constants.}, which is an extension of NL with two involutive negations from BL. InNL is a multiplicative, constant-free part of Nonassociative Bilinear Logic (NBL), which is a nonassociative variant of BL, studied later in Płaczek \cite{placzek-one-sided-2020}. One may easily construct CyNBL -- a cyclic version of NBL, first studied by de Groote and Lamarche in \cite{de-groote-classical-2002}.

To sum up, we have six variants of linear logics due to commutativity and associativity. The same can be applied to Lambek calculi. Noncommutative variants are more useful for modeling natural language syntax (see Abrusci \cite{abrusci-classical-2002} for discussion), while nonassociative variants are better for modeling linguistic tree-structures. That is why all of them are worth studying.

\begin{definition}
Let \(L_1\) and \(L_2\) be two logics. We say that \(L_2\) is a \emph{(weak) conservative extension} of \(L_1\), if every theorem of \(L_2\) expressed in the language of \(L_1\) is also a theorem of \(L_1\). We say that \(L_2\) is a \emph{strong conservative extension} of \(L_1\), if the consequence relation of \(L_2\) restricted to the language of \(L_1\) coincides with the consequence relation of \(L_1\).
\end{definition}

\begin{remark}\label[remark]{rem::translation}
When the language of \(L_1\) is a sublanguage of \(L_2\), ``the language of \(L_1\)'' above is understood syntactically. More generally, as is the case for FNL and CyNBL in this paper, it should be understood via a fixed translation of \(L_1\)--formulas and sequents into \(L_2\). For this to be a reasonable translation to test conservativity against, it should agree with the algebraic semantics of both logics: every model of \(L_2\) should already be a model of \(L_1\), once the connectives of \(L_1\) are read off along the translation. The translation we use, denoted \(\tau\), is defined in \cref{sec::strong-conservativeness} (\cref{dfn::F}); it is exactly the syntactic counterpart of the fact, observed in \cref{dfn::fnlt-algebra,dfn::cynbl-algebra} below, that every CyNBL--algebra is an \FNLt--algebra once we set \(a \multimap b = a^\sim \oplus b\) and \(a \multimapinv b = a \oplus b^\sim\) -- the very definitions \(\tau\) uses for \(\multimap\) and \(\multimapinv\). As a further, purely syntactic, sanity check, \(\tau\) is also injective (\cref{sec::strong-conservativeness}), so distinct \FNLt--formulas are never conflated in CyNBL.
\end{remark}

Abrusci \cite{abrusci-classical-2002} proved that CyBL is not a conservative extension of FL, if we admit additive constants in these logics. But, without additive constants, BL and CyBL are conservative extensions of FL. Płaczek \cite{placzek-sequent-2024,placzek-applying-2026} proved that CyNBL is a strongly conservative extension of FNL without additive constants. For FNL with additive constants we can use a similar argument as in Abrusci \cite{abrusci-classical-2002} to show that CyNBL is not a conservative extension of FNL, even weakly.

The problem is that in CyBL or CyNBL one additive constant can be obtained from the other by negation, namely \(\top = \bot^\sim\) (\cref{dfn::cynbl-algebra}). This leads to the problem that the weaker language of FL (FNL) has a nonempty intersection with negated FL (FNL) formulas in CyBL (CyNBL). In this paper we ask and answer the following question: does conservativity hold if we allow FL (FNL) only one of the two additive constants \(\top\) and \(\bot\), while CyBL (CyNBL) still admits both? Here \(\top\) is the \emph{upper bound}, i.e. the greatest element of the lattice order, and \(\bot\) is the \emph{lower bound}, i.e. its least element (\cref{dfn::fnlt-algebra,dfn::cynbl-algebra}). If we keep only the lower bound \(\bot\) in FNL, the answer is negative, since the upper bound \(\top\) is already definable from \(\bot\) in FNL (as \(\bot \multimap \bot\); see the remark after \cref{dfn::fnlt-algebra}), so the same overlap of languages reappears. But if we keep only the upper bound \(\top\) in FNL, the answer is positive. We show that CyNBL, with both additive constants, is a strongly conservative extension of FNL with only one additive constant -- the upper bound. Moreover, we establish this result in the more general setting of the system \(\mathbf{S}_\Phi\), i.e. in the presence of an arbitrary set \(\Phi\) of nonlogical axioms (assumptions), which strengthens the result considerably.

\subsection{Notational remarks}

In the introduction above we deliberately skipped most of the symbols, since different traditions clash and we must specify the symbols we use; see Paoli \cite{paoli-substructural-2002} for discussion. In this paper we use \(\otimes\) for Lambek product (tensor, fusion) as well as for its linear logic counterpart. We use \(\multimap\) and \(\multimapinv\) for left and right implications (residuations) in Lambek Calculus, while we use \(\oplus\) for the par (coproduct) in linear logic. We use \(1\) for the multiplicative unit for \(\otimes\) in both logics and \(0\) for the multiplicative unit for \(\oplus\) in linear logics. All these connectives and constants are called multiplicative. The additive connectives (lattice operations) are \(\wedge\) and \(\vee\), while additive constants are \(\top\) and \(\bot\). The negation is denoted \(^\sim\) written as a suffix (for shorter formulas) and it is a meta-language operation, not a connective.

FNL stands for Full Nonassociative Lambek Calculus with both additive constants, while \FNLt stands for Full Nonassociative Lambek Calculus with only one additive constant \(\top\).

\subsection{Computer Science interpretation}
\label{sec::cs-interpretation}

The applications of Lambek calculi and linear logics in linguistics and computer science are widely studied; see Morrill \cite{morrill-categorial-2011} for comprehensive treatment. The question which may arise is -- what is the point of considering such an artificial logic as \FNLt. One of the applications of Lambek Calculus in computer science is type resolving in programming. We have a tree of symbols (represented as a sequence with parentheses), we assign types to every symbol and ask the question: what is the type of the result. Every compiler for strongly typed programming languages must do so.

In this context the product is an operation representing the tree structure of function (operation) calls, the residuations are for functions and operations to denote argument types and the resulting types. The lattice operations are for combination of types -- \(A \wedge B\) denotes a symbol of both type \(A\) and \(B\), e.g. an instance of the class which inherits from interfaces \(A\) and \(B\); \(A \vee B\) denotes a symbol which can be of type \(A\) or \(B\), but we do not know which one, e.g. addition of numbers which can take both integers and floats. \(1\) is a type of symbol which is neutral in the code and can be put anywhere, e.g. comments. In this natural interpretation of types, \(\top\) represents a general type of all symbols. For example in C\# it can be \verb+object+. The \(\bot\) would be the opposite -- a placeholder we can put in every context (like \(1\)), but it takes space and is treated like a true function argument or even a true function itself. It is unnatural for any programming language to have such a symbol. In C\# we have the \verb+dynamic+ type, but it cannot replace everything, or in Elm we have the \verb+Debug+ module, but it still cannot be put everywhere. Hence it is natural to consider the logic with only one additive constant -- the upper bound.

\subsection{Algebras}

We will briefly introduce certain algebras that serve as models for the logics examined in this paper. The reason to do so is for better understanding of the logics and for the sake of completeness. We will not use these algebras in the syntactic proofs, but their basic properties can easily show the basic properties of the considered logics.

We use the same symbols for operations and constants in the algebras as for the corresponding connectives and constants in the logics. This is not a real problem, since the context will always make it clear whether we are talking about the logic or about the algebra.

\begin{definition}[\FNLt--algebra]
	\label{dfn::fnlt-algebra}
	Let \(\mathbf{M} = (M, \otimes, \multimap, \multimapinv, \wedge, \vee, 1, \top, \leq)\) be a structure such that \(\otimes, \multimap, \multimapinv, \wedge, \vee\) are binary operations, \(1 \in M\) and \(\leq\) is a partial order on \(M\). We say that \(\mathbf{M}\) is an \emph{\FNLt--algebra}, if the following conditions hold:
	\begin{enumerate}[label={(\roman*)}]
		\item \((M, \wedge, \vee, \leq)\) is a lattice;
		\item \(\top \in M\) and \(a \leq \top\) for all \(a \in M\);
		\item \(a \otimes b \leq c\) iff \(a \leq c \multimapinv b\) iff \(b \leq a \multimap c\) for all \(a,b,c \in M\);
		\item \(1 \otimes a = a = a \otimes 1\) for all \(a \in M\).
	\end{enumerate}
\end{definition}

\FNLt--algebras serve as models of \FNLt. It is possible to extend \FNLt by introducing additive constant \(\bot\), which is the lower bound of the lattice.

Notice that if we introduce \(\bot\), then \(\top\) is definable as \(\bot \multimap \bot\). One shows that \(\bot \otimes a = \bot\) and hence \(a \leq \bot \multimap \bot\) for all \(a\). 

\begin{definition}[De Morgan negation]
	\label{dfn::de-morgan-negation}
	Let \((P, \leq)\) be a poset and let \(^\sim\) be a unary operation on \(P\) such that for all \(a, b \in P\):
	\begin{enumerate}[label={(\roman*)}]
		\item if \(a \leq b\), then \(b^\sim \leq a^\sim\);
		\item \(a^{\sim\sim} = a\).
	\end{enumerate}
	Such an operation \(^\sim\) is called a \emph{De Morgan negation}.
\end{definition}

\begin{definition}[CyNBL--algebra]
	\label[definition]{dfn::cynbl-algebra}
	Let \(\mathbf{M} = (M, \otimes, \wedge, \vee, {^\sim}, 1, \bot, \leq)\) be a structure such that \(\otimes, \wedge, \vee\) are binary operations, \(^\sim\) is a unary operation, 1 and \(\bot\) are constants and \(\leq\) is a partial order on \(M\). We say that \(\mathbf{M}\) is a \emph{CyNBL--algebra}, if the following conditions hold:
	\begin{enumerate}[label={(\roman*)}]
		\item \(^\sim\) is a De Morgan negation;
		\item \((M, \wedge, \vee, \leq)\) is a lattice;
		\item \(a \otimes b \leq c\) iff \(b \otimes c^\sim \leq a^\sim\) iff \(c^\sim \otimes a \leq b^\sim\) for all \(a,b,c \in M\);
		\item \(1 \otimes a = a = a \otimes 1\) for all \(a \in M\);
		\item \(\bot \leq a\) for all \(a \in M\).
	\end{enumerate}
\end{definition}

CyNBL--algebras serve as models of CyNBL. One defines \(a \oplus b = (b^\sim \otimes a^\sim)^\sim\) and \(0 = 1^\sim\) and \(\top = \bot^\sim\).

It can be proved that every CyNBL--algebra is also an \FNLt--algebra. We define \(a \multimap b = a^\sim \oplus b\) and \(a \multimapinv b = a \oplus b^\sim\). One checks that the condition \cref{dfn::fnlt-algebra}(iii) holds.

\subsection{Formalization}
In this work, all statements are proved constructively by structural induction with carefully chosen measures.
To ensure the correctness of the proofs, we formalize the proof systems \FNLt and the auxiliary system for CyNBL and the main results of the paper in the proof assistant Agda~\cite{Bove2009Agda}.
For an online introduction, we refer the readers to the \href{https://agda.readthedocs.io/en/stable/getting-started/what-is-agda.html}{Agda documentation}.
Each of the formalized contents is labelled by \agdahen\ that links to the corresponding line of the code.

This formalization is an extension of the one in Veltri and Wan \cite{VW2025Tableaux}. We utilize the case distinction type to compute all possible cases arising from equations of substituted trees. 
\section{Sequent systems}

In this section we present three sequent systems: one for \FNLt and two for CyNBL. The systems of CyNBL have already been considered in Płaczek \cite{placzek-sequent-2024}, so for the sake of completeness we recall useful properties of these systems without proof. The systems for CyNBL are one--sided (left--sided), while the system for \FNLt is two--sided (intuitionistic).

Both \FNLt and CyNBL have their own languages we define in the respective sections. In the sequents (both one--sided and two--sided) formulas are structured in \emph{bunches} (or \emph{formula-trees}). Bunches are elements of the free unital groupoid generated by the set of formulas. The neutral element of this unital groupoid is called an empty bunch and denoted \(\epsilon\). We represent bunches using capital Greek letters and formulas using capital Latin letters.

Having bunches we must define the \emph{context}. An \emph{anonymous variable} is a unique formula represented as \(\_\), serving as a placeholder for substitution. It is important to note that if a bunch contains multiple anonymous variables, they are considered distinct, even if they share the same symbol. A \emph{context} is a bunch with an anonymous variable. Contexts are denoted by \(\Gamma[\_]\), and when we perform the substitution of \(\Delta\) in place of \(\_\), we represent it as \(\Gamma[\Delta]\).

\subsection{Full Nonassociative Lambek Calculus with \texorpdfstring{\(\top\)}{⊤}}

Let \(\mathscr{V}\) be an arbitrary, countable set of variables. We construct the set of \emph{\FNLt--formulas} from \(\mathscr{V}\) by the binary connectives \(\otimes, \multimap, \multimapinv, \wedge, \vee\) and the constants \(1\) and \(\top\) (\agdaref{Fma.agda\#L163}). We denote the set of all \FNLt--formulas by \(\mathscr{F}_\FNLt\).

An \emph{\FNLt--sequent} is a pair \(\Gamma \Rightarrow A\), where \(\Gamma\) is an \FNLt--bunch and \(A\) is an \FNLt--formula. We denote the set of all \FNLt--sequents by \(\mathscr{L}_\FNLt\).

\renewcommand{\arraystretch}{2}

\noindent The axioms and the rules of \FNLt are as follows (\agdaref{FNL.agda\#L24}):
\begin{rules}
	\text{(id)} & \AxiomC{\(\)}\UnaryInfC{\(A \Rightarrow A\)}\DisplayProof & \FNLcut & \AxiomC{\(\Gamma \Rightarrow A \)}\AxiomC{\( \Delta[A] \Rightarrow C\)}\BinaryInfC{\(\Delta[\Gamma] \Rightarrow C\)}\DisplayProof\\
	& & \FNLRtotop & \AxiomC{\(\)}\UnaryInfC{\(\Gamma \Rightarrow \top\)}\DisplayProof\\
	\FNLotimesR & \AxiomC{\(\Gamma[(A, B)] \Rightarrow C\)}\UnaryInfC{\(\Gamma[A \otimes B] \Rightarrow C\)}\DisplayProof &
	\FNLRtootimes & \AxiomC{\(\Gamma \Rightarrow A \)}\AxiomC{\( \Delta \Rightarrow B\)}\BinaryInfC{\(\Gamma, \Delta \Rightarrow A \otimes B\)}\DisplayProof\\
	\FNLmultimapR & \AxiomC{\(\Gamma[B] \Rightarrow C \)}\AxiomC{\( \Theta \Rightarrow A\)}\BinaryInfC{\(\Gamma[(\Theta, A \multimap B)] \Rightarrow C\)}\DisplayProof &
	\FNLRtomultimap & \AxiomC{\(A, \Gamma \Rightarrow B\)}\UnaryInfC{\(\Gamma \Rightarrow A \multimap B\)}\DisplayProof\\
	\FNLmultimapinvR & \AxiomC{\(\Gamma[A] \Rightarrow C \)}\AxiomC{\( \Theta \Rightarrow B\)}\BinaryInfC{\(\Gamma[(A \multimapinv B, \Theta)] \Rightarrow C\)}\DisplayProof &
	\FNLRtomultimapinv & \AxiomC{\(\Gamma, B \Rightarrow A\)}\UnaryInfC{\(\Gamma \Rightarrow A \multimapinv B\)}\DisplayProof\\
	\FNLoneR & \AxiomC{\(\Gamma[\Delta] \Rightarrow C\)}\UnaryInfC{\(\Gamma[(1, \Delta)] \Rightarrow C\)}\DisplayProof \quad \AxiomC{\(\Gamma[\Delta] \Rightarrow C\)}\UnaryInfC{\(\Gamma[(\Delta, 1)] \Rightarrow C\)}\DisplayProof &
	\FNLRtoone & \AxiomC{\(\)}\UnaryInfC{\(\epsilon \Rightarrow 1\)}\DisplayProof\\
	\FNLveeR & \AxiomC{\(\Gamma[A] \Rightarrow C \)}\AxiomC{\( \Gamma[B] \Rightarrow C\)}\BinaryInfC{\(\Gamma[A \vee B] \Rightarrow C\)}\DisplayProof &
	\FNLRtovee & \AxiomC{\(\Gamma \Rightarrow A_i\)}\UnaryInfC{\(\Gamma \Rightarrow A_1 \vee A_2\)}\DisplayProof(i=1,2)\\
	\FNLwedgeR & \AxiomC{\(\Gamma[A] \Rightarrow C\)}\UnaryInfC{\(\Gamma[A \wedge B] \Rightarrow C\)}\DisplayProof \quad \AxiomC{\(\Gamma[B] \Rightarrow C\)}\UnaryInfC{\(\Gamma[A \wedge B] \Rightarrow C\)}\DisplayProof &
	\FNLRtowedge & \AxiomC{\(\Gamma \Rightarrow A \)}\AxiomC{\( \Gamma \Rightarrow B\)}\BinaryInfC{\(\Gamma \Rightarrow A \wedge B\)}\DisplayProof
\end{rules}

\FNLt is strongly complete with respect to \FNLt--algebras. A valuation is a homomorphism \(\mu\) of a free algebra of \FNLt--formulas to a \FNLt--algebra extended by the following properties: \(\mu((\Gamma, \Delta)) = \mu(\Gamma) \otimes \mu(\Delta)\) and \(\mu(\epsilon) = 1\). A sequent \(\Gamma \Rightarrow A\) is satisfied by a valuation \(\mu\), if \(\mu(\Gamma) \leq \mu(A)\). One proves strong completeness in a standard way, using Lindenbaum--Tarski algebras. Since every CyNBL--algebra is an \FNLt--algebra, CyNBL is an extension of \FNLt.

\subsection{Cyclic Nonassociative Bilinear Logic}

Let \(\mathscr{V}\) be an arbitrary, countable set of variables. We define the set of atoms \(\mathscr{V}'\) as follows: if \(p \in \mathscr{V}\), then both \(p\) and \(p^\sim\) are elements of \(\mathscr{V}'\). Variables in this set are referred to as \emph{positive} atoms, while their negations (\(p^\sim\)) are termed \emph{negative} atoms. We construct the set of \emph{CyNBL--formulas} from \(\mathscr{V}'\) by the binary connectives: \(\otimes\), \(\oplus\), \(\wedge\) and \(\vee\) and the constants 1, 0, \(\top\) and \(\bot\) (\agdaref{Fma.agda\#47}). We denote the set of all CyNBL--formulas by \(\mathscr{F}_\text{CyNBL}\).

It is worth noting that we do not treat negation as a connective. The systems we introduce adhere to the negation normal form, meaning that negations only appear in the form of atoms.

We define the metalanguage negation \({^\sim}\) (\agdaref{Fma.agda\#57}):
\[\begin{array}{cc}
		(p)^\sim = p^\sim                         & (p^\sim)^\sim = p                         \\
		1^\sim = 0                                & 0^\sim = 1                                \\
		\top^\sim = \bot                          & \bot^\sim = \top                          \\
		(A \otimes B)^\sim = B^\sim \oplus A^\sim & (A \oplus B)^\sim = B^\sim \otimes A^\sim \\
		(A \wedge B)^\sim = A^\sim \vee B^\sim    & (A \vee B)^\sim = A^\sim \wedge B^\sim
	\end{array}\]
One notices \(A^{\sim\sim} = A\) for all formulas \(A\).

A \emph{CyNBL--sequent} is defined as any nonempty bunch. We denote the set of all CyNBL--sequents by \(\mathscr{L}_\text{CyNBL}\).

\noindent The axioms of CyNBL are:
\begin{rules}
	\text{(a-id)} & \AxiomC{\(\)}\UnaryInfC{\(p, p^\sim\)}\DisplayProof & \text{(a-0)} & \AxiomC{\(\)}\UnaryInfC{\(\,0\,\)}\DisplayProof\\
	\text{(a-\(\bot\))} & \AxiomC{\(\)}\UnaryInfC{\(\Gamma[\bot]\)}\DisplayProof
\end{rules}
\noindent The introduction rules (rules introducing connectives and constants) are:
\begin{rules}
	\rrtensor & \AxiomC{\(\Gamma[(A,B)]\)}\UnaryInfC{\(\Gamma[A \otimes B]\)}\DisplayProof & \rrone & \AxiomC{\(\Gamma[\Delta]\)}\UnaryInfC{\(\Gamma[(1, \Delta)]\)}\DisplayProof \quad \AxiomC{\(\Gamma[\Delta]\)}\UnaryInfC{\(\Gamma[(\Delta, 1)]\)}\DisplayProof\\
	\roplusone & \AxiomC{\(\Gamma[B]\)}\AxiomC{\( \Delta, A\)}\BinaryInfC{\(\Gamma[(\Delta, A \oplus B)]\)}\DisplayProof & \roplustwo & \AxiomC{\(\Gamma[A]\)}\AxiomC{\( B, \Delta\)}\BinaryInfC{\(\Gamma[(A \oplus B, \Delta)]\)}\DisplayProof\\
	\rwedge & \AxiomC{\(\Gamma[A]\)}\UnaryInfC{\(\Gamma[A \wedge B]\)}\DisplayProof \quad \AxiomC{\(\Gamma[A]\)}\UnaryInfC{\(\Gamma[B \wedge A]\)}\DisplayProof & \rvee & \AxiomC{\(\Gamma[A] \)}\AxiomC{\( \Gamma[B]\)}\BinaryInfC{\(\Gamma[A \vee B]\)}\DisplayProof\\
\end{rules}
In (r-1) we assume \(\Delta \neq \epsilon\).

\noindent The structural rules and the cut rule are:
\begin{rules}
	\rshift & \AxiomC{\((\Gamma, \Delta), \Theta\)}\UnaryInfC{\(\Gamma, (\Delta, \Theta)\)}\DisplayProof & \rcyc & \AxiomC{\(\Gamma, \Delta\)}\UnaryInfC{\(\Delta, \Gamma\)}\DisplayProof\\
	\rcut & \AxiomC{\(\Gamma[A] \)}\AxiomC{\( \Delta, A^\sim\)}\BinaryInfC{\(\Gamma[\Delta]\)}\DisplayProof
\end{rules}

The rules \rshift and \rcyc are reversible. For \rcyc it is obvious. To obtain the reversed \rshift we consecutively apply \rcyc, \rshift, \rcyc, \rshift and again \rcyc. The reversibility of these rules is an important fact used in the omitted proofs.

The models of CyNBL are CyNBL--algebras. A valuation is a homomorphism \(\mu\) of a free algebra of CyNBL--formulas to a CyNBL--algebra extended by the following properties: \(\mu((\Gamma, \Delta)) = \mu(\Gamma) \otimes \mu(\Delta)\) and \(\mu(\epsilon) = 1\). A sequent \(\Gamma\) is satisfied by a valuation \(\mu\), if \(\mu(\Gamma) \leq 0\).

CyNBL is strongly complete with respect to CyNBL--algebras. The \rshift rule expresses the condition \cref{dfn::cynbl-algebra}(iii). The rule \rcyc expresses the fact that we have a De Morgan negation. One proves strong completeness in the usual way, using Lindenbaum--Tarski algebras.

\subsection{Auxiliary system for CyNBL \texorpdfstring{\agdaref{CyNBL.agda\#L68}}{}}

Here we provide a special sequent system, which is equivalent to the system above, but has some special properties. First, we replace structural rules with an axiom and two new rules, because structural rules are hard to consider in the syntactic proofs. Then, we add \emph{assumption rules}, which are special versions of the cut rule, where instead of the cut formula we use nonlogical axioms.

To simplify the reasoning, we assume all the nonlogical axioms (assumptions) are of the form \(C, D^\sim\). This is not a real restriction, since every CyNBL--sequent can be represented in this form. The rule \rrtensor is reversible, so we can replace every bunch by a formula, where we replace the commas with \(\otimes\). If the sequent is a single formula, we can always add \(1\). Hence, every CyNBL--sequent can be represented in the form \(C, D^\sim\).

Let \(\Phi\) be a set of sequents of the form \(C, D^\sim\). We define the system \(\mathbf{S}_\Phi\). The system \(\mathbf{S}_\Phi\) has all axioms and introduction rules of CyNBL. We add the following rules and axioms:
\begin{rules}

	\text{(a-id2)} & \AxiomC{\(\)}\UnaryInfC{\(p^\sim, p\)}\DisplayProof\\
	\roplusthree & \AxiomC{\(A, \Gamma\)}\AxiomC{\( B, \Delta\)}\BinaryInfC{\(A \oplus B, (\Delta, \Gamma)\)}\DisplayProof &
	\roplusfour & \AxiomC{\(\Gamma, A \)}\AxiomC{\( \Delta, B\)}\BinaryInfC{\((\Delta, \Gamma), A \oplus B\)}\DisplayProof
\end{rules}
In \roplusthree and \roplusfour we assume \(\Gamma, \Delta\) are nonempty; otherwise they are special cases of \roplustwo and \roplusone.

\noindent For every \((C, D^\sim) \in \Phi\) we add the \emph{assumption rules}:
\begin{rules}
	\rassmone & \AxiomC{\(\Gamma[D] \)}\AxiomC{\(\Delta, C^\sim\)}\BinaryInfC{\(\Gamma[\Delta]\)}\DisplayProof &
	\rassmtwo & \AxiomC{\(\Gamma[C^\sim] \)}\AxiomC{\(\Delta, D\)}\BinaryInfC{\(\Gamma[\Delta]\)}\DisplayProof 
\end{rules}
From now on we denote by \(\provS \Gamma\) the provability of \(\Gamma\) in \(\mathbf{S}_\Phi\).

The aim of the structural rules \rshift and \rcyc is to allow us to rearrange the bunches in the sequents. Since CyNBL is neither associative nor commutative, this rearrangement allows only moving the outer parenthesis and cyclically permuting the bunches. This rearrangement creates equivalent sequents, i.e. sequents that are provable in \(\mathbf{S}_\Phi\) iff the original ones are provable. The following proposition states this fact.

The following results and their proofs can be found in Płaczek \cite{placzek-sequent-2024}. The first statement originates from a result in Buszkowski \cite{buszkowski-classical-2016}.

We define the relation \(\Gamma \sim \Delta\), which holds for the CyNBL--bunches \(\Gamma\) and \(\Delta\), if \(\Delta\) can be derived from \(\Gamma\) by finitely many applications of \rcyc and \rshift. Since both \rcyc and \rshift are reversible, this relation is an equivalence relation.
\begin{proposition}\label[proposition]{prop::main}
	\begin{enumerate}[label={(\roman*)}]
		\item \label{prop::sim}
			Let \(\Gamma[\_]\) be a CyNBL--context. Then, there exists a unique CyNBL--bunch \(\Delta\) such that \(\Gamma[\_] \sim (\Delta, \_)\).
		\item \label{prop::simcor}
			Let \(\Gamma[\_] \sim (\Delta, \_)\) and let \(\Theta\) be a substructure of \(\Gamma[\_]\), which does not contain this occurrence of \_ (but it can contain occurrences of other anonymous variables). Then, the reduction preserves \(\Theta\).
		\item \label{prop::closedsub}
			The relation \(\sim\) is closed under substitution.
		\item \label{prop::simproof}
			Let \(\Gamma \sim \Delta\). Then \(\Gamma\) is provable in \(\mathbf{S}_\Phi\) iff \(\Delta\) is provable. The height of the proof of \(\Gamma\) is equal to the height of the proof of \(\Delta\) (\agdaref{AdmShiftCyc.agda\#L683}).
		\item \label{prop::admiss}
			The rules \rshift (\agdaref{AdmShiftCyc.agda\#627}) and \rcyc (\agdaref{AdmShiftCyc.agda\#631}) are admissible in \(\mathbf{S}_\Phi\).
	\end{enumerate}
\end{proposition}

The proof of \cref{prop::main} is similar to the one in Płaczek \cite{placzek-sequent-2024}, so we omit it.

The cut rule is also admissible in \(\mathbf{S}_\Phi\). The cut--elimination property for CyNBL is a known fact, but we need a stronger result, which is the admissibility of cut in \(\mathbf{S}_\Phi\) for every \(\Phi\). This result allows us to obtain not only cut elimination for pure CyNBL, but also something like partial cut elimination for CyNBL with nonlogical axioms (assumptions). This result will be useful later.

\begin{lemma}[\agdaref{CutElim.agda}]
	\label[lemma]{lem::cut-elimination}
	The rule \rcut is admissible in \(\mathbf{S}_\Phi\), i.e. if \(\provS \Gamma[A]\) and \(\provS \Delta, A^\sim\), then \(\provS \Gamma[\Delta]\).
\end{lemma}

Again, the proof of this lemma can be found in Płaczek \cite{placzek-sequent-2024}. The proof is a standard induction on the complexity of the cut formula \(A\) and the heights of the proofs of \(\Gamma[A]\) and \(\Delta, A^\sim\). The main idea is to show that if \(A\) is not an axiom, then we can permute the cut rule with the introduction rules, until we obtain cuts with smaller cut formulas.

\begin{proposition}\label[proposition]{prop::Phi=SPhi}
	\(\Phi \vdash \Gamma\) iff \(\provS \Gamma\).
\end{proposition}

Provability in pure CyNBL is equivalent to provability in \(\mathbf{S}_\emptyset\). Hence, CyNBL admits cut elimination.

Now that we know the two systems are equivalent, we can use the system \(\mathbf{S}_\Phi\) as an easier system to work with. One of the standard methods of proving that a logic is a conservative extension of another logic is to show the subformula property. Here we have to deal with a more complicated case -- the languages are not the same (but are translatable) and we have nonlogical axioms. The following proposition states the \emph{generalized} subformula property for \(\mathbf{S}_\Phi\), which is a form of subformula property for logics with an arbitrary set of nonlogical axioms. This result is a consequence of the cut elimination for \(\mathbf{S}_\Phi\).

Let \(T\) be an arbitrary set of CyNBL--formulas. By a \emph{\(T\)--sequent} we mean a CyNBL--sequent containing only formulas from \(T\). By a \emph{\(T\)--proof} we mean a proof consisting of only \(T\)--sequents.

\begin{proposition}[generalized subformula property]
	\label[proposition]{prop::subformula-property}
	Let \(\Phi\) be a set of sequents of the form \(C, D^\sim\). Let \(\Gamma\) be a CyNBL--sequent and let \(T\) be a subformula--closed set such that every formula in \(\Gamma\) occurs in \(T\) and for every \((C, D^\sim) \in \Phi\) we have \(C, D, C^\sim, D^\sim \in T\). Then, \(\Gamma\) is provable in \(\mathbf{S}_\Phi\) iff it has a \(T\)--proof in \(\mathbf{S}_\Phi\).
\end{proposition}

This property is really strong, because it shows that the proof--search procedure needs to check only sequents built from a finite set of formulas. This method is usually enough to show that a logic is a strongly conservative extension of another logic, if they share the same language. The following section will show that this method is also enough to show that CyNBL is a strongly conservative extension of \FNLt, even though they do not share the same language and we have nonlogical axioms.

\section{Strong conservativeness}
\label{sec::strong-conservativeness}

In this section we show that CyNBL is a strongly conservative extension of \FNLt. We define the translation of \FNLt--formulas to CyNBL--formulas and we show that if the translation of an \FNLt--sequent is provable in CyNBL, then the original sequent is provable in \FNLt. The challenge lies in recognizing the translated \FNLt--formulas among all CyNBL--formulas.

The results of this section have been partially proved in Płaczek \cite{placzek-sequent-2024}, but without additive constants. Here we provide the full proof, including the case with additive constants. We will underline the differences between the two cases, which are not trivial.

\begin{definition}[\agdaref{Translations.agda\#L15}]
	\label{dfn::F} We define two sets of CyNBL--formulas: \(\mathcal{F}_1\) and \(\mathcal{F}_2\). The former is the set of \FNLt--formulas translated into CyNBL and the latter is the set of negated translated \FNLt--formulas.
	\begin{enumerate}[label={(\roman*)}]
		\item For every \(p \in \mathcal{V}\), \(p \in \mathcal{F}_1\) and \(p^\sim \in \mathcal{F}_2\).
		\item \(1, \top \in \mathcal{F}_1\) and \(0, \bot \in \mathcal{F}_2\).
		\item If \(A, B \in \mathcal{F}_1\), then \(A \otimes B, A \wedge B, A \vee B \in \mathcal{F}_1\).
		\item If \(A, B \in \mathcal{F}_2\), then \(A \oplus B, A \wedge B, A \vee B \in \mathcal{F}_2\).
		\item If \(A \in \mathcal{F}_1\) and \(B \in \mathcal{F}_2\), then \(A \oplus B, B \oplus A \in \mathcal{F}_1\) and \(A \otimes B, B \otimes A \in \mathcal{F}_2\).
		\item No other formula belongs to \(\mathcal{F}_1\) nor \(\mathcal{F}_2\).
	\end{enumerate}
\end{definition}

Notice that for every \(A \in \mathcal{F}_1\) its metalanguage negation \(A^\sim \in \mathcal{F}_2\) and conversely, if \(A \in \mathcal{F}_2\), then \(A^\sim \in \mathcal{F}_1\). Moreover, \(\mathcal{F}_1 \cap \mathcal{F}_2 = \emptyset\). We define \(\mathcal{F} = \mathcal{F}_1 \cup \mathcal{F}_2\). We see that \(\mathcal{F}\) is not the set of all CyNBL--formulas; e.g. \(p \vee p^\sim \notin \mathcal{F}\).

CyNBL is an extension of \FNLt. We define the translation \(\tau\) of \FNLt--sequents to CyNBL--sequents, inductively as follows (\agdaref{Translations.agda\#L78}):
\[
\begin{array}{@{}l@{\qquad}l@{\qquad}l@{}}
	\tau(p)=p & \tau(1)=1 & \tau(\top)=\top \\
	\tau(A\otimes B)=\tau(A)\otimes\tau(B) & \tau(A\multimap B)=\tau(A)^\sim\oplus\tau(B) & \tau(A\multimapinv B)=\tau(A)\oplus\tau(B)^\sim \\
	\tau(A\wedge B)=\tau(A)\wedge\tau(B) & \tau(A\vee B)=\tau(A)\vee\tau(B) & \tau((\Gamma,\Delta))=(\tau(\Gamma),\tau(\Delta)) \\
	& \tau(\Gamma \Rightarrow A)=(\tau(\Gamma),\;\tau(A)^\sim) &
\end{array}
\]

We see that \(\tau(\mathcal{F}_\FNLt) = \mathcal{F}_1\). As explained in \cref{rem::translation}, \(\tau\) is a reasonable translation to test conservativity against because it agrees with the algebraic semantics of the two logics: on \(\multimap\) and \(\multimapinv\) it uses exactly the operations \(a \multimap b = a^\sim \oplus b\) and \(a \multimapinv b = a \oplus b^\sim\) by which every CyNBL--algebra is already seen to be an \FNLt--algebra (\cref{dfn::fnlt-algebra,dfn::cynbl-algebra}), and on the remaining connectives and constants \(\tau\) is essentially the identity. As a further syntactic property, \(\tau\) is also injective: it is defined by recursion on the structure of \FNLt--formulas, and no two syntactic cases in its definition can produce the same CyNBL--formula (in particular, \(A \multimap B\) and \(A \multimapinv B\) are respectively translated to \(\tau(A)^\sim \oplus \tau(B)\) and \(\tau(A) \oplus \tau(B)^\sim\), which are never confused with one another or with the translation of any other connective). Hence \(\tau\) is a bijection between \(\mathscr{F}_\FNLt\) and \(\mathcal{F}_1\) -- extended, in the obvious way via \(\tau((\Gamma,\Delta))=(\tau(\Gamma),\tau(\Delta))\), to a bijection between \FNLt--bunches and bunches of \(\mathcal{F}_1\)--formulas -- and we write \(\tau^{-1}\) for its inverse. From now on, whenever a bunch or formula built entirely from \(\mathcal{F}_1\) is displayed as an antecedent or succedent of \(\provFNL\) (as, e.g., \(\Delta\) and \(A\) in \cref{thm::main-theorem} below), it stands for its \(\tau^{-1}\)--image, i.e. for the unique \FNLt--bunch or formula that \(\tau\) translates to it.

The translated \FNLt--sequents are \(\mathcal{F}\)--sequents where exactly one \(\mathcal{F}_2\)--formula appears at the end. But, the structural rules of CyNBL allow us to rearrange the formulas in the sequents, so we will work on \(\mathcal{F}\)--sequents, which contain exactly one \(\mathcal{F}_2\)--formula, but it can be anywhere in the sequent.

In Płaczek \cite{placzek-sequent-2024} we showed that a provable \(\mathcal{F}\)--sequent must contain at most one \(\mathcal{F}_2\)--formula. Then, we used \cref{prop::main} to transform the proofs of these sequents into proofs in \FNLt. Here it is not true. Admitting additive constants, especially \(\bot\), allows CyNBL to prove some \(\mathcal{F}\)--sequents with more than one \(\mathcal{F}_2\)--formula. For example, the sequent \(p, p^\sim, \bot\) is provable in CyNBL, but it contains two \(\mathcal{F}_2\)--formulas: \(p^\sim\) and \(\bot\). Here, we complete the results of Płaczek \cite{placzek-sequent-2024} by showing that if an \(\mathcal{F}\)--sequent is provable in \(\mathbf{S}_\Phi\), then it must contain at least one formula from \(\mathcal{F}_2\), which is valid even without additive constants.

\begin{lemma}[\agdaref{Lemma3-2.agda}]
	\label[lemma]{lemma::unprov}
	Let \(\Phi\) be a set of sequents of the form \(C, D^\sim\), where \(C, D \in \mathcal{F}_1\). If an \(\mathcal{F}\)--sequent is provable in \(\mathbf{S}_\Phi\), then it must contain at least one formula from \(\mathcal{F}_2\).
\end{lemma}

\begin{proof}
	Let \(\Phi\) be a set of sequents of the form \(C, D^\sim\), where \(C, D \in \mathcal{F}_1\). Let \(\Gamma\) be an \(\mathcal{F}\)--sequent, provable in \(\mathbf{S}_\Phi\). We show that \(\Gamma\) must contain at least one formula from \(\mathcal{F}_2\). By the generalized subformula property, every provable \(\mathcal{F}\)--sequent must have an \(\mathcal{F}\)--proof.
	
	We proceed by induction on the length of the proof of \(\Gamma\) in \(\mathbf{S}_\Phi\). We notice that every axiom contains at least one formula from \(\mathcal{F}_2\).

	We assume that \(\Gamma\) is the conclusion of a rule. We consider the cases:
	\begin{enumerate}[label={(\roman*)}]
		\item Let \(\Gamma\) be the conclusion of \rrtensor. Then \(\Gamma = \Gamma'[A \otimes B]\) for some \(\Gamma'[\_]\) and \(A, B\). By the induction hypothesis, there must be at least one formula from \(\mathcal{F}_2\) in \(\Gamma'[A, B]\). We consider cases:
			\begin{enumerate}[label={(\alph*)}]
				\item If there is a formula from \(\mathcal{F}_2\) in \(\Gamma'\), then it is also in \(\Gamma'[A \otimes B]\).
				\item If \(A \in \mathcal{F}_2\), then \(A \otimes B \in \mathcal{F}_2\) and hence, it is in \(\Gamma'[A \otimes B]\).
				\item If \(B \in \mathcal{F}_2\), then \(A \otimes B \in \mathcal{F}_2\) and hence, it is in \(\Gamma'[A \otimes B]\).
			\end{enumerate}
		\item The cases for \rvee, \rwedge and (\rrone) are analogous.
		\item Let \(\Gamma\) be the conclusion of \roplusone. Then \(\Gamma = \Gamma'[(\Delta, A \oplus B)]\) for some \(\Gamma'[\_], \Delta, A, B\). By the induction hypothesis, there must be at least one formula from \(\mathcal{F}_2\) in both \(\Gamma'[B]\) and \(\Delta, A\). If \(A \in \mathcal{F}_1\), then \(\Delta\) must contain a formula from \(\mathcal{F}_2\), and hence, \(\Gamma\) contains a formula from \(\mathcal{F}_2\). If \(B \in \mathcal{F}_1\), then \(\Gamma'[\_]\) must contain a formula from \(\mathcal{F}_2\), and hence, \(\Gamma\) contains a formula from \(\mathcal{F}_2\). Let \(A, B \in \mathcal{F}_2\). Then \(A \oplus B \in \mathcal{F}_2\) and hence, it is in \(\Gamma\).
		\item The cases for \roplustwo, \roplusthree and \roplusfour are analogous.
		\item We consider assumption rules. We will show the case for \rassmone, the rest is analogous. Let \(\Gamma = (\Gamma_1, \Gamma_2)\) and \((C, D^\sim) \in \Phi\). By the induction hypothesis, both \(D, \Gamma_2\) and \(\Gamma_1, C^\sim\) contain formulas from \(\mathcal{F}_2\). Since \(D \in \mathcal{F}_1\), \(\Gamma_2\) must contain a formula from \(\mathcal{F}_2\), and hence, \(\Gamma\) contains a formula from \(\mathcal{F}_2\).
	\end{enumerate}
\end{proof}

This result, despite the fact that it is only a helpful lemma, shows how adding \(\bot\) and its axiom changes the properties of the system. This difference is the reason why we must exclude \(\bot\) from the language of \FNLt, if we want to show that CyNBL is a strongly conservative extension of \FNLt. If we allowed \(\bot\) in \FNLt, then the translation of \FNLt--sequents would contain only one formula from \(\mathcal{F}_2\), which is the translation of the succedent, but CyNBL would be able to prove some sequents with more than one formula from \(\mathcal{F}_2\), and we would not be able to recognize the translated \FNLt--sequents among all provable \(\mathcal{F}\)--sequents.

\begin{theorem}[\agdaref{Conservativeness.agda\#L179}]
	\label[theorem]{thm::main-theorem}
	Let \(\Gamma[A^\sim]\) be an \(\mathcal{F}\)--sequent where \(A \in \mathcal{F}_1\) and \(A^\sim\) is the only \(\mathcal{F}_2\)--formula in \(\Gamma[A^\sim]\), and \(\Gamma[A^\sim] \sim \Delta, A^\sim\) for some \(\Delta\). Then \(\provS \Gamma[A^\sim]\), if and only if \(\Phi \provFNL \tau^{-1}(\Delta) \Rightarrow \tau^{-1}(A)\).
\end{theorem}

Below and for the remainder of the paper we drop \(\tau^{-1}\) from the notation and simply write \(\Delta \Rightarrow A\), relying on the convention already fixed above: \(\Delta\) and \(A\) are \(\mathcal{F}_1\)--bunches and formulas, identified with the \FNLt--bunches and formulas they are the \(\tau\)--image of. Writing out \(\tau^{-1}\) at every step of the induction below would only obscure the argument, so we take it as understood from this point on.

\begin{proof}
	The if-part immediately follows from \cref{prop::Phi=SPhi} and the fact that CyNBL is an extension of \FNLt. We prove only the only-if-part.

	Let \(\Theta[C^\sim]\) be an \(\mathcal{F}\)--sequent provable in \(\mathbf{S}_\Phi\) and \(\Theta[\_] \sim (\Psi, \_)\). We show the construction of the proof in \FNLt. We proceed by the outer induction on the number of connectives of \(C\) and the inner induction on the proof of \(\Theta[C^\sim]\). Notice that both \(\Theta[\_]\) and \(\Delta\) have only \(\mathcal{F}_1\)--formulas.
	
	\begin{enumerate}[label={\arabic*\(^\circ\)}]
	
	\item Let \(C = p\). We run the inner induction. If \(\Theta[p^\sim]\) is an axiom, then \(\Theta[p^\sim] = (p, p^\sim)\) or \(\Theta[p^\sim] = (p^\sim, p)\). In both cases, \((\Psi, C^\sim) = (p, p^\sim)\) and \(p \Rightarrow p\) is an axiom in \FNLt.

	\(\Theta[p^\sim]\) cannot be an instance of (a-\(\bot\)), because \(\bot\) is not in \(\mathcal{F}_1\).

	Now we assume \(\Theta[p^\sim]\) is the conclusion of a rule. We observe \(p^\sim\) cannot be the active formula of a rule.

		We consider \rrtensor. The premise is \(\Gamma[(A, B)][p^\sim]\). Then, by \cref{prop::main}\ref{prop::simproof}, \(\provS \Delta[(A,B)], p^\sim\) for \((\Psi, p^\sim) = (\Delta[A \otimes B], p^\sim)\). By the inner induction hypothesis, \(\Delta[(A,B)] \Rightarrow p\) is provable in \FNLt from \(\Phi\). We apply \FNLotimesR and obtain \(\Delta[A \otimes B] \Rightarrow p\). This is the only case, since \(p^\sim \in \mathcal{F}_2\).

	The cases for \rvee, \rwedge and \rrone are analogous.

	We consider \roplusone. Since we assume \(A \oplus B \in \mathcal{F}_1\), either \(A \in \mathcal{F}_1\) or \(B \in \mathcal{F}_1\). We consider the following cases:
	\begin{gather*}
		\AxiomC{\(\Gamma[B][p^\sim] \)}\AxiomC{\( \Delta, A\)}\BinaryInfC{\(\Gamma[(\Delta, A \oplus B)][p^\sim]\)}\DisplayProof \qquad
		\AxiomC{\(\Gamma[B] \)}\AxiomC{\( \Delta[p^\sim], A\)}\BinaryInfC{\(\Gamma[(\Delta[p^\sim], A \oplus B)]\)}\DisplayProof\\
	\end{gather*}

	In the first case, \(A = D^\sim \in \mathcal{F}_2\), by \cref{lemma::unprov}, and \(B \in \mathcal{F}_1\); otherwise \(\Delta\) would contain a formula from \(\mathcal{F}_2\), which contradicts the assumption. By the inner induction hypothesis, \(\Phi \provFNL \Gamma'[B] \Rightarrow p\) and \(\Phi \provFNL \Delta \Rightarrow D\). We apply \FNLmultimapR and obtain \(\Phi \provFNL \Gamma'[(\Delta, D \multimap B)] \Rightarrow p\). Also, \((\Psi, p^\sim) = (\Gamma'[(\Delta, D \multimap B)], p^\sim)\) by \cref{prop::main}\ref{prop::simcor}. We notice that \(\tau(D \multimap B) = A \oplus B\).
	
	In the second case, \(B = D^\sim \in \mathcal{F}_2\), by \cref{lemma::unprov} and \(A \in \mathcal{F}_1\); otherwise \(\Gamma[\_]\) would contain a formula from \(\mathcal{F}_2\), which contradicts the assumption. We know that \(\Gamma[B]\) reduces to \(\Gamma', D^\sim\) and, by the inner induction hypothesis, \(\Phi \provFNL \Gamma' \Rightarrow D\). Also, \(\Delta[p^\sim], A\) reduces to \(\Delta'[A], p^\sim\). So, \(\Phi \provFNL \Delta'[A] \Rightarrow p\). We apply \FNLmultimapinvR and obtain \(\Phi \provFNL \Delta'[(A \multimapinv D, \Gamma')] \Rightarrow p\). One easily checks that \(\Gamma[(\Delta[p^\sim], A \multimapinv D)] \sim (\Delta'[(A \multimapinv D, \Gamma')], p^\sim)\).

	The cases for \roplustwo, \roplusthree and \roplusfour are similar.

	We consider \rassmone. We have two possibilities.
	\[
		(1)\, \AxiomC{\(F, \Gamma[p^\sim] \)}\AxiomC{\( \Delta, E^\sim\)}\BinaryInfC{\(\Delta, \Gamma[p^\sim]\)}\DisplayProof,
	\]
	where \((\Delta, \Gamma[p^\sim]) = \Theta[p^\sim]\) and \((\Delta, \Gamma[p^\sim]) \sim (\Gamma'[\Delta], p^\sim)\) and \(\Gamma'[\Delta] = \Psi\).

		By the inner induction hypothesis, \(\Phi \provFNL \Delta \Rightarrow E\) and \(\Phi \provFNL \Gamma'[F] \Rightarrow p\). Since \(E \Rightarrow F \in \Phi\), \(\Phi \provFNL E \Rightarrow F\). We apply \FNLcut to \(\Gamma'[F] \Rightarrow p\) and \(E \Rightarrow F\) and obtain \(\Phi \provFNL \Gamma'[E] \Rightarrow p\). We apply \FNLcut to this and \(\Delta \Rightarrow E\) and obtain \(\Phi \provFNL \Gamma'[\Delta] \Rightarrow p\).

	\[
		(2)\, \AxiomC{\(D, \Gamma \)}\AxiomC{\( C^\sim, \Delta[p^\sim]\)}\BinaryInfC{\(\Delta[p^\sim], \Gamma\)}\DisplayProof,
	\]
	where \((\Delta[p^\sim], \Gamma) = \Theta[p^\sim]\) and \(\Delta'[\Gamma] = \Psi\). This case is impossible, since, by \cref{lemma::unprov}, the first premise is unprovable.

	The cases for other assumption rules are similar.

	\item Let \(C = 1\). We run the inner induction. \(\Theta[0]\) may be an axiom (a-0). So \(\Theta[\_] = \_\) and \(\Psi = \epsilon\). Hence, \((\Psi, C^\sim) = 1^\sim\) and \(\Phi \provFNL \epsilon \Rightarrow 1\) by \FNLRtoone.

	We assume \(\Theta[0]\) is not an axiom. Then, \(0\) is not the active formula of any rule. We proceed as above.

	\item Let \(C = \top\). We run the inner induction. \(\Theta[\bot]\) may be an axiom (a-\(\bot\)). By \cref{prop::main}\ref{prop::simcor}, \(\Theta[\bot] \sim (\Theta', \bot)\). Since we assumed that \(C^\sim\) is the only \(\mathcal{F}_2\)--formula in \(\Theta[C^\sim]\), then \(\Theta'\) consists of \(\mathcal{F}_1\)--formulas. Hence, \((\Psi, C^\sim) = (\Theta', \bot)\) and \(\Phi \provFNL \Theta' \Rightarrow \top\) by \FNLRtotop.

	We assume \(\Theta[\bot]\) is not an axiom. Then, \(\bot\) is not the active formula of any rule. We proceed as above.

	\item Let \(C = A \otimes B\). Then \(C^\sim = B^\sim \oplus A^\sim\). We run the inner induction. We know that \(\Theta[C^\sim]\) cannot be an axiom. So it has to be the conclusion of a rule. If \(C^\sim\) is not the active formula of a rule, we proceed as above. We assume \(C^\sim\) is the active formula, so we have the following possibilities:
	\begin{gather*}
		(1)\, \AxiomC{\(\Gamma[A^\sim] \)}\AxiomC{\( \Delta, B^\sim\)}\BinaryInfC{\(\Gamma[(\Delta, B^\sim \oplus A^\sim)]\)}\DisplayProof \qquad (2)\, \AxiomC{\(\Gamma[B^\sim] \)}\AxiomC{\( A^\sim, \Delta\)}\BinaryInfC{\(\Gamma[(B^\sim \oplus A^\sim, \Delta)]\)}\DisplayProof\\
		(3)\, \AxiomC{\(B^\sim, \Gamma \)}\AxiomC{\( A^\sim, \Delta\)}\BinaryInfC{\(B^\sim \oplus A^\sim, (\Delta, \Gamma)\)}\DisplayProof \qquad (4)\, \AxiomC{\(\Gamma, B^\sim \)}\AxiomC{\( \Delta, A^\sim\)}\BinaryInfC{\((\Delta, \Gamma), B^\sim \oplus A^\sim\)}\DisplayProof
	\end{gather*}
	In (1) we have \(\Gamma[\_] \sim (\Gamma', \_)\), so, \(\Gamma[(\Delta, B^\sim \oplus A^\sim)] \sim ((\Gamma', \Delta), B^\sim \oplus A^\sim)\). By the inner induction hypothesis, \(\Phi \provFNL \Gamma' \Rightarrow A\) and \(\Phi \provFNL \Delta \Rightarrow B\). We apply \FNLRtootimes and obtain \(\Phi \provFNL \Gamma', \Delta \Rightarrow A \otimes B\). In (2), (3) and (4) we proceed similarly.

	\item Let \(C = A \oplus B\). Then either \(A \in \mathcal{F}_1\) or \(B \in \mathcal{F}_1\). Suppose \(B \in \mathcal{F}_1\). Then \(C = D^\sim \oplus B\) and \(C^\sim = B^\sim \otimes D\). We run the inner induction. \(\Theta[C^\sim]\) cannot be an axiom. Hence it has to be the conclusion of a rule. If \(C^\sim\) is not the active formula of a rule, we proceed as above. We assume \(C^\sim\) is the active formula, so we have only one possibility:
	\[
		\AxiomC{\(\Gamma[(B^\sim, D)]\)}\UnaryInfC{\(\Gamma[B^\sim \otimes D]\)}\DisplayProof
	\]
	Then \(\Gamma[\_] \sim (\Gamma', \_)\) and \(\Gamma[(B^\sim, D)] \sim ((D, \Gamma'), B^\sim)\). By the induction hypothesis, \(\Phi \provFNL D, \Gamma' \Rightarrow B\). We apply \FNLRtomultimap and obtain \(\Phi \provFNL \Gamma' \Rightarrow D \multimap B\). We notice that \(\tau(D \multimap B) = D^\sim \oplus B\).

	The other case is similar.

	\item Let \(C = A \vee B\), then \(C^\sim = A^\sim \wedge B^\sim\). We run the inner induction. \(\Theta[C^\sim]\) cannot be an axiom. Hence it has to be the conclusion of a rule. If \(C^\sim\) is not the active formula of a rule, we proceed as above. We assume \(C^\sim\) is the active formula, so we have two possibilities:
	\[
		\AxiomC{\(\Theta[A^\sim]\)}\UnaryInfC{\(\Theta[A^\sim \wedge B^\sim]\)}\DisplayProof \qquad \AxiomC{\(\Theta[B^\sim]\)}\UnaryInfC{\(\Theta[A^\sim \wedge B^\sim]\)}\DisplayProof
	\]
	In both cases we apply the inner induction hypothesis to the premise and then we apply \FNLRtowedge.

	\item Let \(C = A \wedge B\), then \(C^\sim = A^\sim \vee B^\sim\). We run the inner induction. \(\Theta[C^\sim]\) cannot be an axiom. Hence it has to be the conclusion of a rule. If \(C^\sim\) is not the active formula of a rule, we proceed as above. We assume \(C^\sim\) is the active formula, so we have one case:
	\[
		\AxiomC{\(\Theta[A^\sim] \)}\AxiomC{\( \Theta[B^\sim]\)}\BinaryInfC{\(\Theta[A^\sim \vee B^\sim]\)}\DisplayProof
	\]
	We apply the inner induction hypothesis to both premises and then we apply \FNLRtovee.
	
	\end{enumerate}
\end{proof}

\begin{corollary}
	CyNBL is a strongly conservative extension of \FNLt.
\end{corollary}

\section{Final remarks}

Abrusci \cite{abrusci-classical-2002} provides a counterexample to the conservativeness of CyBL over FL (Theorem 10), leaving the proof to the reader. We point out that this example contains a small gap: translated literally into the language of this paper, Abrusci's sequent is not in fact provable in CyNBL, nor even in CyBL. A minor correction to the sequent restores the intended counterexample; we detail this below.

Abrusci's sequent, translated to the language of this paper, is:
\[
	\left(p \multimap q\right) \multimap \bot \, , \, p \multimap \left(\left( \bot \multimap r\right) \multimap s\right) \Rightarrow s
\]
Abrusci states that this sequent is not provable in FL (and hence not in FNL), but is provable in CyBL (and in CyNBL). Its translation to CyNBL is:
\[
	\left(\left(q^\sim \otimes p\right) \oplus \bot \, , \, p^\sim \oplus \left(\left(r^\sim \otimes \bot\right) \oplus s\right)\right) \, , \, s^\sim
\]
This is the gap: the displayed sequent is not provable in CyNBL, nor even in CyBL. It is easily fixed:
\[
	p \multimap \bot \, , \, \left(q \multimap \top\right) \multimap \left(\left( \bot \multimap r\right) \multimap s\right) \Rightarrow s
\]
Now, it is provable in CyNBL, but not in FNL. In FNL we would need to have \(\Rightarrow p\). In CyNBL:
\[
\AxiomC{}
\LeftLabel{(a-\(\bot\))}
\UnaryInfC{\(r^\sim, \bot\)}
\LeftLabel{\rrtensor}
\UnaryInfC{\(r^\sim \otimes \bot\)}
\AxiomC{}
\LeftLabel{(a-id)}
\UnaryInfC{\(s, s^\sim\)}
\LeftLabel{\roplustwo}
\BinaryInfC{\(\left(r^\sim \otimes \bot\right) \oplus s, s^\sim\)}
\AxiomC{}
\LeftLabel{(a-\(\bot\))}
\UnaryInfC{\(p^\sim, \left(\bot , q\right)\)}
\LeftLabel{\rrtensor}
\UnaryInfC{\(p^\sim, \left(\bot \otimes q\right)\)}
\AxiomC{}
\LeftLabel{(a-\(\bot\))}
\UnaryInfC{\(\bot\)}
\LeftLabel{\roplustwo}
\BinaryInfC{\(p^\sim \oplus \bot, \left(\bot \otimes q\right)\)}
\LeftLabel{\roplusone}
\BinaryInfC{\( \left(\left(p^\sim \oplus \bot\right) \, , \, \left(\left(\bot \otimes q\right) \oplus \left(\left(r^\sim \otimes \bot\right) \oplus s\right)\right)\right) \, , \, s^\sim \)}
\DisplayProof
\]

\paragraph*{Acknowledgements}
C.-S.\ Wan was supported by Estonian Research Council grants no.\ PSG749.

The authors want do declare the scope of individual contributions to this paper. P.\ P\l{}aczek was responsible for the main results of the paper, including the proof of strong conservativeness. C.-S.\ Wan was responsible for the Agda formalization of the results and for the counterexample to the conservativeness of CyBL over FL.

\bibliographystyle{eptcs}
\bibliography{fullbib}

\end{document}